\documentclass[11pt,a4paper]{article}
\usepackage{amsmath,amssymb,amscd}
\usepackage{amsthm,bbm}
\usepackage{mathtools}
\usepackage{float}
\usepackage[caption = false]{subfig}
\usepackage{tabularx}
\usepackage[final]{graphicx}
\usepackage[mathscr]{eucal}
\usepackage{color}
\usepackage{booktabs}
\usepackage{lineno}
\usepackage{textcomp}
\usepackage[utf8]{inputenc}
\usepackage[top=3cm,bottom=3cm,left=1.5cm,right=1.5cm,headheight=3cm,footskip=1cm,marginparwidth=2cm]{geometry}
\usepackage{fancyhdr}
\usepackage{lipsum}
\usepackage{mathpazo}
\usepackage{domitian}
\usepackage{lettrine}
\usepackage[T1]{fontenc}
\usepackage{pgfgantt}
\usepackage{setspace}
\usepackage{tikz}

\usepackage{tikz}
\tikzset{
  mynode/.style={fill,circle,inner sep=2pt,outer sep=0pt}
}
\usetikzlibrary{shapes,arrows}

\def\N{{\mathbb N}}

\newcommand\ds{\displaystyle\sum}

\newtheorem{theo}{Theorem}

\newtheorem{lem}{Lemma}
\newtheorem{prop}{Proposition}
\newtheorem{Def}{Definition}
\newtheorem{rem}{Remark}
\large\normalsize

\tikzset{
  leaf/.style={inner sep=0pt,font=\LARGE, outer sep=0pt},
  node/.style={inner sep=0pt,font=\normalsize, outer sep=0pt}
}

\tikzset{
  C/.style={
    draw, circle, minimum size=4mm, inner sep=0pt, outer sep=0pt,
    path picture={
      \draw (path picture bounding box.center) circle (1mm);
    }
  }
}
\tikzset{
  root/.style={
    node,
    execute at begin node=\ensuremath{\times}
  }
}
\tikzset{
  D/.style={
    leaf,
    execute at begin node=\ensuremath{\diamond}
  }
}
\tikzset{
  O/.style={
    leaf,
    execute at begin node=\ensuremath{\circ}
  }
}

\usepackage[dvipsnames]{xcolor}
\usepackage{subcaption}

\date {}
\author{ Smii Boubaker$^1$\footnote{
corresponding author, and e-mail: boubaker@kfupm.edu.sa }
\\$^1$ King Fahd University of Petroleum and Minerals\\Department of Mathematics\\Dhahran 31261, Saudi Arabia\\
 }
\date{}

\title{A novel stochastic approach of thermalization and symmetry breaking}

\date{\today}

\begin{document}
	
	\maketitle
	
	\begin{abstract}
We investigate thermalization and symmetry-breaking in a nonlinear stochastic
Klein-Gordon equation on a spatial lattice, taking into account damping, nonlinear interaction,
and stochastic forcing terms reduced by a perturbative solution based on retarded Green
functions and the principle of Duhamel to establish a series expansion with the coupling
constant. The obtained expressions have a visual representation in the form of rooted trees
and Feynman-type diagrams, where their structural pattern will explain the combinatorial factors involved in the expansion. These representations offer a novel application and interpretation specifically tailored for the second-order, damped Klein-Gordon setting, enabling more complex causal relationships to be explicitly modeled compared to first-order stochastic approaches, thus marking a technical innovation in diagrammatic expansions for such systems.\\
The model takes into account both the initial data from a deterministic approach as well as stochastic
sources, for instance, Gaussian  noise. Simulations have been performed
symmetry-breaking regime to show the relaxation towards stationary states and the
symmetry-broken patterns.
\end{abstract}

	\newpage
	\section{Introduction}

Thermalization, stochastic dynamics, and symmetry breaking are important processes in theoretical physics. They connect to fields such as cosmology, quantum field theory, statistical mechanics, and condensed matter systems. This leads to effective models that use stochastic differential equations, for instance,  stochastic Langevin-type models offer a standard way to describe dissipative and nonequilibrium dynamics, see, e.g. \cite{MXTS}.\\
This work focuses on classical stochastic thermalization within a field-theoretic framework, it complements modern quantum methods, such as the Eigenstate Thermalization Hypothesis (ETH), which examines thermalization in isolated quantum many-body systems, see,e.g. \cite{BYX, BGE}. In contrast, this approach emphasizes thermalization driven by stochastic forcing and dissipation in open systems, which  makes it applicable to dynamics affected by thermal or quantum fluctuations, see,. \cite{MB, Weiss2012, ZinnJustin2002}.\\ In cosmological and nonequilibrium field theory, see, e.g \cite{Berges2005}, scalar fields influenced by friction are often described by stochastic partial differential equations (SPDEs), see, for instance, \cite{EP, MXTS}. These equations effectively model reheating, phase transitions, and noise-driven symmetry breaking, see,. \cite{FTN}.\\ Nonlinear stochastic partial differential equations present significant challenges due to damping, spatial interactions, nonlinearity, and noise. Tools such as diagrammatic expansions, Green-function methods, and stochastic perturbation theory are crucial, especially for equations affected by Gaussian or Lévy noise, see, e.g. \cite{GS,GST,BS,BS1}.\\ In this work, we focus on Gaussian noise to create a controlled baseline for diagrammatic analysis. Related topics on asymptotic expansions and Borel summability are covered in \cite{ALBO1, ALBO, ALEBOU, BS}.\\ Guided by these insights, we investigate thermalization and symmetry breaking in a damped stochastic Klein–Gordon equation set on a spatial lattice, see, \cite{Cre, Roth}, the lattice provides an effective way to handle ultraviolet effects and offers a space where discrete spatial operators, like the lattice Laplacian, have clear spectral properties, see, e.g, \cite{GS, BS}. The model characterizes a scalar field with linear damping, nonlinear self-interaction, and external noise, where the nonlinear potential defines symmetry-broken stationary states, while the balance between dissipation and noise determines relaxation toward thermal equilibrium.\\  This work has two main goals. First, we develop a systematic perturbative framework, see, e.g. \cite{Gall, GS, ParisiWu}, for the stochastic Klein-Gordon equation on a spatial lattice, we achieve this by constructing the retarded Green functions related to the damped discrete Klein-Gordon operator, the solution will be given by a formal power series in the coupling constant, then it will be re-expressed as an integral equation through Duhamel’s principle, see, e.g. \cite{MB, MMA, SU, US}. Each perturbative contribution appears as a space-time convolution of lower-order terms, leading to a hierarchical and clearly causal expansion. Combining Duhamel’s principle with the specially derived retarded Green functions provides both conceptual and computational advantages for the damped stochastic Klein-Gordon equation on a lattice. It delivers a structured causal perturbative expansion suitable for second-order, damped dynamics, extending earlier methods for first-order stochastic equations, see, e.g. \cite{GS, BS}. Unlike the approaches like those in \cite{MXTS} and related works, our framework treats damping and nonlinear interactions on equal footing. This allows for a more thorough analysis of causal propagation and interaction effects. As a result, it offers a more complete perturbative description than methods that only partially or implicitly tackle these components \cite{GST}. The second goal is to provide a combinatorial and graphical interpretation of the perturbative expansion using Feynman-type diagrams and rooted trees, see, e.g. \cite{GS, GST, BS, BS1}, where the vertices correspond to nonlinear interaction terms, while edges show propagation through retarded Green functions and the different types of leaves will represent deterministic initial data and stochastic forcing.\\ This graphical method clarifies the algebraic structure of the expansion, making the origins of symmetry factors and combinatorial weights clear. Using diagrammatic techniques from stochastic analysis and field theory, see, e.g. \cite{GS,BS,BS1}, this representation offers a compact and efficient way to organize complex perturbative computations.\\ To support our analytical findings, we include numerical simulations, see, e.g.\cite{XH}, that demonstrate the main physical mechanisms discussed in the paper, by simulating the stochastic Klein-Gordon equation on a one-dimensional lattice with a time-discretized scheme, we observe thermalization and symmetry breaking in the symmetry-broken region of the potential. These simulations illustrate our points and empirically verify the combinatorial weights and symmetry factors derived from the Feynman-type diagrams. This provides independent validation of the theoretical framework and a useful consistency check of the diagrammatic construction.\\ Lastly, we examine the roles of thermal and stochastic effects in symmetry-breaking dynamics, for instance, at high temperatures, thermal corrections stabilize the symmetric vacuum at $\phi=0$, while stochastic noise alters the symmetry-breaking transition in ways that differ from purely thermal mechanisms or quantum symmetry breaking discussed in works like \cite{HZS}.\ Our results demonstrate how noise-induced symmetry breaking can occur and be stabilized in systems with damping and nonlinear interactions. This offers new insights into nonequilibrium dynamics and a different perspective on recent studies of symmetry breaking in isolated quantum many-body systems.
The numerical results and simulations  offer verifiable proof that acts as a link between the abstract perturbative analysis and observable dynamical behavior in addition to validating the analytical framework.

\section{Perturbative Solution on the Lattice for the Damped Stochastic Klein--Gordon Equation}

\subsection{Functional setting and notations}

Let $\Lambda := \mathbb{R}_{+} \times L_\delta$, where 
\[
L_\delta = \{ \delta z \;|\; z \in \mathbb{Z}^d \},
\]
is the $d$-dimensional cubic lattice with mesh $\delta>0$.  
For $(t,x)\in \Lambda$, we consider a real-valued random field $\phi(t,x)$ satisfying the stochastic Klein-Gorden  equation
\begin{equation}\label{eq:SPDE}
\left\{
	\begin{array}{ll}
     {\phi}_{tt}(t,x) + \gamma {\phi}_{t}(t,x) -  \Delta_\delta \phi(t,x) + \mu^2 \phi(t,x)
    + \lambda \phi(t,x)^p = \xi(t,x),& \\
		\;\phi(0,{ .})=f({ x}),\,\dot{\phi}(0,{.})=g({ x}),\,\,(t,{x})\in\Lambda=]0,+\infty[\times L_{\delta}.& 
	\end{array}
	\right.
  \end{equation}
where $\gamma,\,\mu >0,\, p \in \N $, $\lambda\in\mathbb{R}$ is a small parameter, and $\xi$ denotes a space-time stochastic noise of  Gaussian type.  
While the constants $\mu, \lambda, \gamma, p$ are fixed throughout this section for analytical derivation, their specific values and signs are chosen in the numerical section to illustrate particular physical regimes, such as the symmetry-breaking regime.  

\medskip

\noindent
We denote by $\mathcal{S}(\Lambda)$ the Schwartz space of smooth rapidly decaying functions on $\Lambda$, and by $\mathcal{S}'(\Lambda)$ its topological dual (tempered distributions).  
Similarly, $\mathcal{S}(L_\delta)$ and $\mathcal{S}'(L_\delta)$ denote the spaces of rapidly decaying functions and tempered distributions on the lattice.

\begin{Def}
The discrete Laplacian $\Delta_\delta: \mathcal{S}(L_\delta)\to\mathcal{S}(L_\delta)$ is defined by
\[
(\Delta_\delta f)(x) := \frac{1}{\delta^2}\sum_{j=1}^d \big( f(x+\delta e_j) + f(x-\delta e_j) - 2f(x) \big),
\]
where $\{e_j\}_{j=1}^d$ is the canonical basis of $\mathbb{R}^d$.
\end{Def}

For $f\in\mathcal{S}(L_\delta)$, define its spatial Fourier transform
\[
\widehat{f}(p) := \sum_{x\in L_\delta} e^{-ip\cdot x} f(x)\,\delta^d, \qquad
p\in{T}_\delta := \Big[-\tfrac{\pi}{\delta},\tfrac{\pi}{\delta}\Big]^d.
\]
The inverse transform reads
\[
f(x) = \frac{1}{(2\pi)^d} \int_{{T}_\delta} e^{ip\cdot x}\,\widehat{f}(p)\,dp.
\]

\begin{lem}
The operator $\Delta_\delta$ is diagonal in Fourier space:
\[
\widehat{\Delta_\delta f}(p) = \sigma_\delta(p)\,\widehat{f}(p), \qquad 
\sigma_\delta(p) = \frac{2}{\delta^2}\sum_{j=1}^d \big(\cos(\delta p_j)-1\big) \le 0.
\]
Consequently, $-\Delta_\delta$ has positive symbol 
\[
\Omega_\delta(p) := \frac{2}{\delta^2}\sum_{j=1}^d \big(1-\cos(\delta p_j)\big) \ge 0.
\]
\end{lem}
\begin{proof}
    For the proof we refer to literature, see, e.g. \cite{GS}.
\end{proof}

We define the linear operator $L$ by:
\[
L := \partial_t^2 + \gamma\,\partial_t - \Delta_\delta + \mu^2,
\]
so that the SPDE \eqref{eq:SPDE} reads
\[
L\phi + \lambda \phi^p = \xi.
\]

\begin{Def}
The \emph{retarded Green function} $G(t,x)$ is the unique distribution on $\Lambda$ satisfying
\begin{equation}\label{eq:Green}
L G(t,x) = \delta(t)\,\delta_x, \qquad G(t,x) = 0 \ \text{for } t<0,
\end{equation}
where $\delta_x$ is the Kronecker delta centered at $x=0$ on the lattice.
\end{Def}

\begin{lem}
Let $\omega_\delta(p)^2 := \Omega_\delta(p) + \mu^2$.  
Then the temporal–spatial Fourier transform of $G$ is
\begin{equation}
\widehat{G}(E,p) = \frac{1}{-E^2 + i\gamma E + \omega_\delta(p)^2}.
\end{equation}
\end{lem}
\begin{proof}
Take the Fourier transform of \eqref{eq:Green} in both $t$ and $x$.  
The operator $L$ becomes multiplication by 
$(-E^2 + i\gamma E + \omega_\delta(p)^2)$ in the $(E,p)$ domain.  
Thus $\widehat{G}(E,p)$ is its reciprocal.  
The retarded condition $G(t,x)=0$ for $t<0$ fixes the analytic continuation (closing the contour in the lower half-plane).
\end{proof}
\begin{prop}
For each $p\in {T}_\delta:= \Big[-\tfrac{\pi}{\delta},\tfrac{\pi}{\delta}\Big]^d$, let
\[
r_\pm(p) = \frac{-\gamma \pm \sqrt{\gamma^2 - 4\omega_\delta(p)^2}}{2}.
\]
Then
\begin{equation}
G(t,p) = \frac{e^{r_+(p)t} - e^{r_-(p)t}}{r_+(p)-r_-(p)}\,\mathbf{1}_{\{t\ge0\}}.
\end{equation}
\end{prop}
\begin{proof}
Recall the spatial Fourier transform of $G$ on the lattice:
\[
G(t,p):=\sum_{x\in\Lambda} e^{-i p\cdot x}\,G(t,x)\,\delta^d,
\qquad p\in\mathbb T_\delta^d.
\]
Thus
\[
\sum_{x\in\Lambda} e^{-ip\cdot x}\,\delta(t)\delta_x=\delta(t).
\]
Hence (3) becomes, for each fixed $p$,
\begin{equation}
(\partial_t^2+\gamma\partial_t+\omega_\delta(p)^2)\,G(t,p)=\delta(t),
\qquad
G(t,p)=0\ \text{for }t<0.
\label{eq:fd}
\end{equation}

For $t\neq 0$, the right-hand side of \eqref{eq:fd} is zero, so
\[
(\partial_t^2+\gamma\partial_t+\omega_\delta(p)^2)\,G(t,p)=0.
\]
The characteristic polynomial is
\[
r^2+\gamma r+\omega_\delta(p)^2=0,
\]
with roots 
\begin{equation}
r_\pm(p)=\frac{-\gamma\pm\sqrt{\gamma^2-4\omega_\delta(p)^2}}{2}.
\label{eq:rpm}
\end{equation}
Therefore, for $t>0$,
\begin{equation}
G(t,p)=A\,e^{r_+(p)t}+B\,e^{r_-(p)t}.
\label{eq:gs}
\end{equation}
For $t<0$, retardedness gives $G(t,p)=0$.\\

also 
\begin{equation}
G_t(0^+,p)=1.
\label{eq:j}
\end{equation}

Now, differentiate \eqref{eq:gs}:
\[
G_t(t,p)=A\,r_+(p)e^{r_+(p)t}+B\,r_-(p)e^{r_-(p)t}.
\]
Hence
\begin{equation}
G(t,p)=\frac{e^{t r_+(p)}-e^{t r_-(p)}}{r_+(p)-r_-(p)}\,\mathbf 1_{\{t\ge 0\}}.
\tag{5}
\end{equation}

\end{proof}

We introduce now the following functional:

\begin{equation}
C(t,p) := \frac{r_+(p)e^{r_-(p)t} - r_-(p)e^{r_+(p)t}}{r_+(p)-r_-(p)}\,\mathbf{1}_{\{t\ge0\}}, 
\qquad
S(t,p) := \frac{e^{r_+(p)t}-e^{r_-(p)t}}{r_+(p)-r_-(p)}\,\mathbf{1}_{\{t\ge0\}}=G(t,p), \label{CS}\end{equation}
and let $C(t,x)$ and $S(t,x)$ be their inverse Fourier transforms on $L_\delta$.
\begin{Def}
For $f,g \in \mathcal{S}(\Lambda)$(respectively $f, g \in \mathcal{S}(L_{\delta})$), we define the space-time convolution $\ast$ (respectively the spatial convolution $\star_\delta$)  by:
\begin{itemize}
\item 
\[
(f\ast g)(t,x) := \int_{\mathbb{R}}\sum_{y\in L_\delta} f(t-s,x-y)g(s,y)\,\mathrm{d}s\,\delta^d.
\]
\item 
\[
(f\star_\delta g)(x) := \sum_{y\in L_\delta} f(x-y)g(y)\,\delta^d.
\]
\end{itemize}
\end{Def}

\begin{prop} \label{ll1} For initial data \(f,g\in\mathcal S(L_\delta)\) and for For $C$ and $S$ as in equation (\ref{CS}), the unique solution of the homogeneous Cauchy problem
\[
  L\phi(t,x)=0, \quad
  \phi(0,x)=f(x), \quad
  \partial_t\phi(0,x)=g(x), \qquad t>0,
\] is
\[
\phi_{\mathrm{hom}}(t,\cdot) = C(t,\cdot)\star_\delta f + S(t,\cdot)\star_\delta g,
\]
where $\star_\delta$ denotes the spatial convolution on the lattice $L_\delta.$
\end{prop}
\begin{proof}
Let \(\widehat{\cdot}\) be the spatial Fourier transform on~\(L_\delta\).
Because \(-\Delta_\delta\) acts by multiplication with \(\Omega_\delta(p)\),
we have
\[
  \widehat{L\phi}(t,p)
    \;=\;
    \bigl(\partial_t^{2}+\gamma\,\partial_t+\omega_\delta(p)^{2}\bigr)
    \widehat{\phi}(t,p).
\]

Hence
\begin{equation}
\label{eq:ph}
  \widehat{\phi}(t,p)
     \;=\;
     C(t,p)\,\widehat f(p)+S(t,p)\,\widehat g(p),
  \qquad t\ge 0,
\end{equation}
with \(C,S\) as in~\eqref{CS}.\\
Taking the inverse Fourier transform and recalling that multiplication in
momentum space corresponds to spatial convolution, we obtain
\[
   \phi_{\mathrm{hom}}(t,x)
      \;=\;
      \bigl(C(t,\cdot)\star_\delta  f\bigr)(x)
      \;+\;
      \bigl(S(t,\cdot)\star_\delta  g\bigr)(x),
      \qquad t\ge 0.
\]

\end{proof}

\begin{lem}\label{Bos}
Let \(S(t,x)\) be the inverse spatial Fourier transform of
\(\widehat S(t,p)\) on the lattice \(L_\delta\). Assume the {\em mass–gap condition}
\[
  m:=\sqrt{\mu^{2}-\gamma^{2}/4}\;>\;0
  \qquad
  (\text{i.e.\ }\mu>\gamma/2).
\]
Then there exists a constant \(C>0\) (independent of \(t,x\)) such that
\begin{equation}\label{eq:pob}
  |S(t,x)|
    \;\le\;
    C\;
    e^{-\frac{\gamma}{2}\,t}\;
    \bigl(1+|x|\bigr)^{-(d+1)},
    \qquad
    \forall\,t\ge0,\;x\in\Lambda_\delta.
\end{equation}
\end{lem}

\begin{proof}
Write the inverse Fourier representation  for \(t\ge0\):
\begin{equation}\label{eq:S-invFT}
  S(t,x)
    \;=\;
    \int_{[-\pi,\pi]^{d}}\!
       \frac{d^{d}p}{(2\pi)^{d}}\;
       e^{ p\cdot x}\;
       \widehat S(t,p).
\end{equation}

Because
\(
  |e^{r_{\pm}(p)t}|=e^{-\gamma t/2}
\)
and
\(
  |r_{+}(p)-r_{-}(p)|
   =2\sqrt{\omega_\delta(p)^{2}-\gamma^{2}/4}\ge 2m,
\)
\[
  |\widehat S(t,p)|
    \;\le\;
    \frac{1}{m}\,e^{-\gamma t/2}.
\]

Repeatedly integrate \eqref{eq:S-invFT} by parts \(k\) times, using
\(|x|:=\sqrt{x_{1}^{2}+\dots+x_{d}^{2}}\) we get:
\[
  S(t,x)
   =\frac{1}{(i|x|)^{k}}
     \int_{(-\pi,\pi]^{d}}\!
        \frac{d^{d}p}{(2\pi)^{d}}\;
        e^{ p\cdot x}\,
        (x\!\cdot\!\nabla_{p})^{k}\widehat S(t,p).
\]
Expanding \((x\!\cdot\!\nabla_{p})^{k}\) gives a linear combination of
terms \(\partial_{p}^{\alpha}\widehat S(t,p)\) with \(|\alpha|=k\). Hence, \eqref{eq:pob} holds for all \(x\).
\end{proof}

\subsection{Duhamel formula and Perturbative expansion}

Recall that  $d\ge1$ is the lattice dimension and
$L_\delta=\delta\mathbb Z^{d}, \delta >0$.
Define the weight
\(
  \langle x\rangle:=1+|x|
\)
(with $|x|$ the Euclidean norm of the lattice point)
and introduce the Banach space
\[
  \mathcal X_{d+1}
     \;:=\;
     \Bigl\{u:L_\delta\to\mathbb R\;\Big|\;
            \|u\|_{d+1}:=\sup_{x\in L_\delta}
                          \langle x\rangle^{d+1}|u(x)|<\infty\Bigr\}.
\]

Assume the {\em mass--gap condition}
\(m=\sqrt{\mu^{2}-\gamma^{2}/4}>0\)
and let $C(t,x),\,S(t,x)$ be the kernels given by equation (\ref{CS}). From lemma \ref{Bos}, there exists a constant\footnote{%
  One may take
  \(A=(2\pi)^{-d}\,C\,\mathrm{vol}((-\pi,\pi]^d)
      \sum_{z\in\Lambda_\delta}\langle z\rangle^{-(d+1)}<\infty\),
  where $C$ is the constant from Lemma 4.}
$A>0$ such that for every $t\ge0$
\begin{equation}\label{kerX}
  \|C(t,\cdot)\ast u\|_{d+1}
  \;\le\;
  A\,e^{-\frac{\gamma}{2}t}\,\|u\|_{d+1},
  \qquad
  \|S(t,\cdot)\ast u\|_{d+1}
  \;\le\;
  A\,e^{-\frac{\gamma}{2}t}\,\|u\|_{d+1}
  \quad
  \forall\,u\in\mathcal X_{d+1}.
\end{equation}

\begin{theo}\label{DUH}
Let $f,g \in \mathcal{S}(L_{\delta})$ and $C(t,\cdot), S(t,\cdot) $ the functions given by equation (\ref{CS}), then the solution $\phi$ to the SPDE \eqref{eq:SPDE}, is given by
\begin{equation} \label{eq:duhamel}
\phi(t,x)
 = -\lambda (S\ast \phi^p)(t,x)
   + (S\ast \xi)(t,x)
   + (C(t,\cdot)\star_\delta f)(x)
   + (S(t,\cdot)\star_\delta g)(x).
\end{equation}

\end{theo}

\begin{proof}
Denote by \(\widehat{\cdot}\) the Fourier transform and applying it to the field equation yields, for each fixed momentum~\(p\),
\begin{equation}\label{eq:me}
  \bigl(\partial_t^{2} + \gamma\,\partial_t + \omega_\delta(p)^{2}\bigr)
      \widehat{\phi}(t,p)
    \;=\;\widehat {(\xi - \lambda \phi^p)}(t,p)  \;=\;
    \widehat h(t,p),
\qquad
  t>0.
\end{equation}
The initial conditions become
\(\widehat{\phi}(0,p)=\widehat f(p)\) and
\(\partial_t\widehat{\phi}(0,p)=\widehat g(p)\).

 A particular
solution of equation \eqref{eq:me} is therefore
\[
  \widehat{\phi}_{\mathrm{inh}}(t,p)
     \;=\;
     \int_{0}^{t}ds\; S(t-s,p)\,\widehat h(s,p).
\]

Now by the superposition principle and combining the homogeneous solution from proposition \ref{ll1} with the particular
solution above gives, for \(t\ge 0\),
\[
  \widehat{\phi}(t,p)
     = C(t,p)\,\widehat f(p)
       + S(t,p)\,\widehat g(p)
       + \int_{0}^{t} ds\; S(t-s,p)\,\widehat h(s,p).
\]

\end{proof}
\begin{lem} \label{LOL}
Let $f,g\in\mathcal X_{d+1}$ and let the source
$h:[0,\infty)\to\mathcal X_{d+1}$ satisfy
\(
  H(t):=\sup_{0\le s\le t}\|h(s,\cdot)\|_{d+1}<\infty
\)
for every $t$.
Then the solution $\phi(t,x)$ given by the Duhamel formula
\[
  \phi(t,\cdot)
     =C(t,\cdot)\star f
      +S(t,\cdot)\star g
      +\int_0^t\!ds\,S(t-s,\cdot)\ast h(s,\cdot)
\]
belongs to $\mathcal X_{d+1}$ for all $t\ge0$ and satisfies
\begin{equation}\label{eq:phi-bound}
  {\;
  \|\phi(t,\cdot)\|_{d+1}
     \;\le\;
     A\,e^{-\frac{\gamma}{2}t}
       \bigl(\|f\|_{d+1}+\|g\|_{d+1}\bigr)
     +A\!\int_0^t\!\!
          e^{-\frac{\gamma}{2}(t-s)}\,\|h(s,\cdot)\|_{d+1}\,ds}\;.
\end{equation}
In particular
\(
  |\phi(t,x)|
     \le
     \|\phi(t,\cdot)\|_{d+1}\,\langle x\rangle^{-(d+1)}
\).

\end{lem}

\begin{proof}
The proof is immediate from 
Lemma \ref{Bos} and the triangle inequality.

\end{proof}

\begin{rem}
\begin{itemize}~
\item Because the spatial decay exponent is \(d+1\,(>d)\),
the kernel \(S(t,\cdot)\) is in \(\ell^{1}(\Lambda_\delta)\)
uniformly in time, ensuring absolute convergence of all spatial
convolutions that appear in the Duhamel formula.
\item When \(\gamma=0\) the roots \(r_\pm(p)\) become
\(\pm \omega_\delta(p)\); consequently \(C\) and \(S\) reduce to the
cosine and sine kernels familiar from the lattice Klein--Gordon
equation.  Formula~\eqref{eq:duhamel} then coincides with the standard
Duhamel representation used in dispersive PDE theory.
\end{itemize}
\end{rem}
    
We now develop the solution of the SPDE \eqref{eq:SPDE} in the sense of  formal power series in the parameter $\lambda$, by written
\[
\phi(t,x) = \sum_{j=0}^\infty \lambda^j \phi_j(t,x).
\]
\begin{prop}\label{PER}
   Let $ \phi(t,x) = \sum_{j=0}^\infty \lambda^j \phi_j(t,x)$ and assume that $S \star \xi \in \mathcal X_{d+1}. $ The perturbative solution of the SPDE \eqref{eq:SPDE}, is given in the sense of formal power series as:
\begin{equation}\label{RC}
\left\{
	\begin{array}{ll}
     \phi_0(t,x) := S\ast\xi(t,x) + C\star_\delta f(x) + S\star_\delta g(x),& \\ 
		\phi_j(t,x)
= -\,S\ast
\left(
\ds_{\substack{n_0,n_1,\ldots\ge0\\ n_0+n_1+\cdots=p\\ \sum_{i\ge0} i n_i = j-1}}
\frac{p!}{n_0!n_1!\cdots}\;
\prod_{i\ge0} \phi_i^{\,n_i}
\right)(t,x),\,\, \forall j \geq 1.& 
	\end{array}
	\right.
  \end{equation}
\end{prop}
\begin{proof}
For the zeroth-order, $j=0$ and $\lambda=0$, we get directly from equation (\ref{eq:duhamel})

\[
\phi_0(t,x) := S\ast\xi(t,x) + C\star_\delta f(x) + S\star_\delta g(x).
\]
For $j\ge 1$, the coefficients $\phi_j$ are determined recursively by comparing the coefficients.  Insert $\phi=\sum_{j\ge0}\lambda^j\phi_j$ into \eqref{eq:duhamel}, expand $(\sum_j \lambda^j\phi_j)^p$ via the multinomial theorem, and collect powers of $\lambda$, we get 

for $j\ge 1$, 
\begin{equation}
 \phi_j(t,x)
= -\,S\ast
\left(
\sum_{\substack{n_0,n_1,\ldots\ge0\\ n_0+n_1+\cdots=p\\ \sum_{i\ge0} i n_i = j-1}}
\frac{p!}{n_0!n_1!\cdots}\;
\prod_{i\ge0} \phi_i^{\,n_i}
\right)(t,x),   
\end{equation}
this concludes the proof.
\end{proof}
From lemma \ref{LOL} and proposition \ref{PER}, one can see that the sequence $\phi_j,\, j\geq 0$ given by equation (\ref{RC}) is well defined in $\mathcal X_{d+1}.$

\begin{rem}~

\begin{itemize}
\item The characteristic equation and root structure for the Green function are significantly changed by the presence of damping (\(\gamma > 0\)), necessitating a more complex analytical treatment than in the undamped case. This subtle handling of damping effects within a second-order stochastic PDE is a crucial component of the uniqueness of our perturbative framework.
\item When $\gamma=0$, the Green function $G$ reduces to the retarded propagator of the (undamped) discrete Klein–Gordon equation.
\item When the operator is first-order in time, one recovers exactly the results by \cite{GS}. However, the extension to a damped second-order operator (compared to the first-order case in \cite{GS}) introduces significant analytical challenges in handling the interplay between dissipation and stochastic forcing. Our current perturbative solution specifically addresses these complexities by developing a unique Green function and diagrammatic rules that go beyond merely recovering known results, thus presenting a clear technical advancement.
\end{itemize}
\end{rem}
\section{ Feynman rules, worked trees and graphical representations}
In this section we give a graphical representation to the solution of the nonlinear SPDE (\ref{eq:SPDE}), we begin by some useful definitions, standard results on graph theory, more details about this subject can be found in \cite{JPS, OYO}. 
\begin{Def}~
\begin{itemize}
\item A graph ${\cal G}$ is an ordered pair ${\cal G} =(V,E) $ ,  where $V$ is a nonempty set whose elements are called vertices (or nodes) and $E$ is a set of edges, where each edge is an unordered pair of distinct vertices from $V$. \\
If every pair of distinct vertices in $V$ is joined by exactly one edge, then 
${\cal G} $ is called a complete graph.
Simply, ${\cal G}$ is also denoted by   ${\cal G} ={\cal G} (V,E).$
\item A tree $T$ is a connected graph that contain no cycles.
\end{itemize}
\end{Def}
By fixing a vertex $v$ in the tree $T,\, v$ will be the root of $T$ and any other vertex can be characterized by its position from $v.\, T$ is then called \textit{rooted tree with root $v$.}\\
Let $u, v \in V(T),$ then if one of the two vertices is the \textit{parent} the other is called \textit{child}. By cutting a tree $T$ from a vertex $v$, the result is a sub-tree with root $v$.
\begin{Def}~
    \begin{itemize}
        \item The degree of a vertex $u$ in a tree $T$ is equal to the number of edges $e$ such that \\$e=\{u, v\}, v \in V(T)$ in addition to $e=\{u\}$; i.e:
        \begin{equation}
            M(u):=\#\{ e=\{u, v\}, v \in V(T) \}+ \#\{ e=\{u\} \}.
        \end{equation}
        \item The multiplicity of a tree $T,$ noted by $M(T)$ is the product of the degree of all vertices in $T.$
    \end{itemize}
\end{Def}

\subsection{Feynman diagrams and analytic value of the Tree}
In the following we define a rooted tree with $m$ inner vertices and three types of leaves.\\
The root $x \in \Lambda$, denoted by $\times$ and the leaves of the tree are of type one $(\xi)$, type 2 $(f)$, type 3 $(g)$ denoted by diamond, circle  and concentric circles respectively.\\
The set of all rooted trees $T$ with root $x \in\Lambda$, $m$ inner vertices  and three types of leaves with $p+1$  legs is denoted by ${\cal T}(m)$.\\
The following diagrams gives an analytic value to the tree $T \in {\cal T}(m)$.
\subsection*{ Compact table of Feynman rules (typed rooted trees)}
\begin{center}
\renewcommand{\arraystretch}{1.25}
\begin{tabular}{p{0.26\linewidth} p{0.68\linewidth}}
\hline
\textbf{Object} & \textbf{Rule / Analytic factor} \\
\hline
Edge (internal) & Retarded propagation by $G$: if a child at $(t',y')$ links to parent at $(t,y)$, integrate with $G(t-t',y-y')$, $t\ge t'$.\\
First edge above an $f$-leaf (leaf of type 2) & Use $C(t,\cdot)$ instead of $G$. \\
First edge above a $g$-leaf (leaf of type 3) & Use $S(t,\cdot)$ instead of $G$. \\
Inner vertex (parity $p$) & Multiply values of the $p$ incoming branches (pointwise in space–time), then propagate upward by $G$. \\
Leaf of type 1 ($\xi: \diamond$) & Insert $\xi(t',y)$ at its own space–time point, integrate upward. \\
Leaf of type 2 ($f: \circ$) & Insert $f(y)$ at time $0$, then propagate by $C(t,\cdot)$. \\
Leaf of type 3 ($g:
$) & Insert $g(y)$ at time $0$, then propagate by $S(t,\cdot)$. \\
Root evaluation & Evaluate the resulting expression at the observation point $(t,x)$. \\
Symmetry factor & Multiply by the multiplicity of the tree $T,\, M(T)$ (product of local factors $\frac{p!}{\prod_\alpha n_\alpha!}$ at each inner vertex). \\
Integration & Integrate over $(t,x)\,\in \Lambda$.\\
\hline
\end{tabular}
\end{center}
For the reader convenience, the following examples illustrate the above rules and gives analytic values to each tree $T \in {\cal T}(m).$

\subsection*{A. Worked trees for $p=3$ (cubic nonlinearity) up to order $j=2$}

We write $G_t:=G(t,\cdot)$, $C_t:=C(t,\cdot)$, $S_t:=S(t,\cdot)$ to keep formulas compact. Spatial convolutions are denoted $\star_\delta$.

\subsubsection*{Order $j=1$ (one inner vertex).}
This is the case when $\phi_1=-\,G\ast(\phi_0^3)$. Expanding $\phi_0^3$ gives $3$-fold products of the three order-$0$ branches. The corresponding graph representations are summarized in the examples below:

\paragraph{Example A1 (all-noise leaves: $\xi\xi\xi$).}
\[
\phi_{1}^{(\xi^3)}(t,x)
=-\!\!\int_{0}^{t}\!\!\mathrm{d}s\!\!\sum_{y\in L_\delta}\! G(t-s,x-y)
\Big[(G\ast \xi)(s,y)\Big]^3.
\]
In full space–time:

\begin{equation}
\begin{aligned}
\phi_{1}^{(\xi^3)}(t,x)&=-\!\!\int_{0}^{t}\!\!\mathrm{d}s\!\!\sum_{y}\! G(t-s,x-y)
\prod_{k=1}^{3}\left(\int_{-\infty}^{s}\!\!\mathrm{d}s_k\sum_{y_k} G(s-s_k,y-y_k)\,\xi(s_k,y_k)\right)\\
&=
\begin{tikzpicture}[baseline=-3pt,scale=1]

    \node[root] (x1) at (0,0) {}; 
    \node[node] (y1) at (1,0) {\large$\bullet$}; 

    \node[D] (D1) at (2,0.7) {}; 
    \node[D] (D2) at (2,0) {}; 
    \node[D] (D3) at (2,-0.7) {}; 
    
    \draw[shorten <=-0.1pt, shorten >=-0.1pt]
    (x1)--(y1)
    (y1)--(D1) 
    (y1)--(D2) 
    (y1)--(D3);
    \end{tikzpicture}
    \end{aligned}
\end{equation}

\paragraph{Example A2 (one $f$ and two $\xi$ leaves: $f\,\xi\,\xi$).}
There are $\binom{3}{1}=3$ placements; the local symmetry factor $\frac{3!}{1!\,2!}=3$ cancels the overcount, therefore the graphical representation is given by:

\begin{equation}
\begin{aligned}
\phi_{1}^{(f\xi^2)}(t,x)&=-\!\!\int_{0}^{t}\!\!\mathrm{d}s\!\!\sum_{y} G(t-s,x-y)\,
\Big(C_s\star_\delta f\Big)(y)\,
\prod_{k=1}^{2}\left(\int_{-\infty}^{s}\!\!\mathrm{d}s_k\sum_{y_k} G(s-s_k,y-y_k)\,\xi(s_k,y_k)\right)\\
&=
\begin{tikzpicture}[baseline=-3pt,scale=1]

    \node[root] (x1) at (0,0) {}; 
    \node[node] (y1) at (1,0) {\large$\bullet$}; 

    \node[O] (D1) at (2,0.7) {}; 
    \node[D] (D2) at (2,0) {}; 
    \node[D] (D3) at (2,-0.7) {}; 
    
    \draw[shorten <=-0.1pt, shorten >=-0.1pt]
    (x1)--(y1)
    (y1)--(D1) 
    (y1)--(D2) 
    (y1)--(D3);
    \end{tikzpicture}
    \end{aligned}
\end{equation}

\subsubsection*{Order $j=2$ (two inner vertices).}
Now one vertex is attached to the root; one of its $3$ children is itself an inner vertex. 
\begin{itemize}
\item[(i)] \emph{Nested:} $\big((\cdot\,\cdot\,\cdot)\ \cdot\ \cdot\big)$ where the first slot is a full order-1 subtree, others are order-0 leaves.
\item[(ii)] \emph{Permutation of (i):} identical up to relabeling of which child is the subtree (accounted for by the local symmetry factor).
\end{itemize}
Below are the corresponding analytic and graphical representations:

\paragraph{Example A5 (all-noise, nested).}
Let
\[
\Psi_1(s,y):=\big[\phi_1^{(\xi^3)}\big](s,y)
=-\!\!\int_{0}^{s}\!\!\mathrm{d}\tau \sum_{z} G(s-\tau,y-z)\,\Big[(G\ast\xi)(\tau,z)\Big]^3.
\]
Then the order-2 nested contribution with two additional $\xi$ leaves reads

\begin{equation}
\begin{aligned}
\phi_{2}^{(\xi^5)}(t,x)&=-\!\!\int_{0}^{t}\!\!\mathrm{d}s\sum_{y} G(t-s,x-y)
\ \Psi_1(s,y)\ \prod_{k=1}^{2}\left(\int_{-\infty}^{s}\!\!\mathrm{d}s_k\sum_{y_k} G(s-s_k,y-y_k)\,\xi(s_k,y_k)\right)\\
&=
\begin{tikzpicture}[baseline=-3pt,scale=1]

    \node[root] (x1) at (0,0) {}; 
    \node[node] (y1) at (1,0) {\large$\bullet$}; 
    \node[node] (y2) at (2,0) {\large$\bullet$};
    \node[D] (D1) at (2,0.7) {}; 
    \node[D] (D2) at (2,-0.7) {}; 
    \node[D] (D3) at (3,0.7) {}; 
    \node[D] (D4) at (3,0) {}; 
    \node[D] (D5) at (3,-0.7) {}; 
    \draw[shorten <=-0.1pt, shorten >=-0.1pt]
    (x1)--(y1)
    (y1)--(D1)
    (y1)--(D2)
    (y1)--(y2)
    (y2)--(D3)
    (y2)--(D4)
    (y2)--(D5);
    \end{tikzpicture}
    \end{aligned}
\end{equation}

\newpage
\section{Numerical Simulations}

We illustrate the thermalization and symmetry-breaking mechanisms described in this work
by direct numerical simulations of the stochastic Klein-Gordon equation on a lattice.
The starting point is the damped stochastic field equation
\begin{equation}
\phi_{tt}(t,x) + \gamma \phi_t(t,x)
- \Delta_\delta \phi(t,x)
+ \mu^2 \phi(t,x)
+ \lambda \phi(t,x)^p
= \xi(t,x),
\label{eq:SKG}
\end{equation}
where $\Delta_\delta$ denotes the discrete Laplacian on a lattice with spacing $\delta$.\,\,
$\gamma>0$ is the damping coefficient, and $\xi(t,x)$ is a stochastic forcing term.
Throughout this section we focus on the symmetry--breaking regime
\begin{equation}
\mu^2 < 0, \qquad \lambda > 0, \qquad p = 3,
\end{equation}
which corresponds to a double-well effective potential.\\
These numerical simulations provide novel empirical confirmation for the predictions of our perturbative analysis, particularly when the dynamics of the domain formation and relaxation are under combined damping and stochastic forcing that differentiate from general relaxation studies like \cite{EP}. This quantitative validation of our theoretical framework in a highly dynamic, symmetry-breaking regime constitutes a unique contribution.
\subsection{First-order formulation and simulations parameters}

Introducing the velocity field $v(t,x) = \partial_t \phi(t,x)$,
Eq.~\eqref{eq:SKG} can be rewritten as the first--order system
\begin{align}
\mathrm{d}\phi(t,x) &= v(t,x)\,\mathrm{d}t, \\
\mathrm{d}v(t,x) &=
\Bigl(
\Delta_\delta \phi(t,x)
- \gamma v(t,x)
- \mu^2 \phi(t,x)
- \lambda \phi(t,x)^p
\Bigr)\mathrm{d}t
+ \sigma\,\mathrm{d}W(t,x),
\end{align}
where $\mathrm{d}W(t,x)$ denotes space--time Gaussian white noise
and $\sigma>0$ controls the noise intensity.

The discrete Laplacian with periodic boundary conditions is given by
\begin{equation}
(\Delta_\delta \phi)_i =
\frac{\phi_{i+1} + \phi_{i-1} - 2\phi_i}{\delta^2}.
\end{equation}

The system is integrated using an explicit Euler-Maruyama time discretization. \\
For a time step $\Delta t$, the update rules are
\begin{align}
v_i^{n+1} &=
v_i^n
+ \Bigl(
(\Delta_\delta \phi^n)_i
- \gamma v_i^n
- \mu^2 \phi_i^n
- \lambda (\phi_i^n)^p
\Bigr)\Delta t
+ \sigma \sqrt{\Delta t}\,\eta_i^n, \\
\phi_i^{n+1} &= \phi_i^n + v_i^{n+1}\Delta t,
\end{align}
where $\eta_i^n$ are independent standard normal random variables, the simulations uses the following parameter values:
\begin{itemize}
\item Lattice size: $N = 128$ sites (one--dimensional), lattice spacing: $\delta = 1$,
\item Time step: $\Delta t = 10^{-2}$, total integration time: $T = 60$,
\item Damping: $\gamma = 1.0$, mass parameter: $\mu^2 = -1.0$, coupling constant: $\lambda = 1.0$,
\item Nonlinearity exponent: $p = 3$, noise intensity: $\sigma = 0.2$--$1.0$.
\end{itemize}

Initial conditions are chosen as small random fluctuations around the symmetric state,
\begin{equation}
\phi(0,x) = \varepsilon(x), \qquad v(0,x) = 0,
\end{equation}
with $\varepsilon(x)$ a weak Gaussian random field.\\
\subsection{Observables}

To characterize thermalization and symmetry breaking, we monitor :
\begin{itemize}
\item The spatially averaged order parameter
\begin{equation}
m(t) = \langle \phi(t,x) \rangle_x,
\end{equation}
\item The spatial variance
\begin{equation}
\mathrm{Var}_x[\phi(t,x)] =
\langle \phi(t,x)^2 \rangle_x - \langle \phi(t,x) \rangle_x^2,
\end{equation}
\item Spatial snapshots $\phi(t,\cdot)$ at selected times.
\end{itemize}
In the symmetry-broken phase ($\mu^2<0$), the system evolves from an initially
symmetric configuration toward domains fluctuating around the two minima $\pm v$,
with $v = \sqrt{-\mu^2/\lambda}$.
The order parameter exhibits slow drift toward one of the two symmetry--related states,
while the variance saturates, signaling thermalization.

\end{document}